\documentclass[a4paper,UKenglish]{lipics-v2016}
\usepackage[noend]{algpseudocode}
\usepackage{xcolor}

\newcommand{\cor}{{\bf corr}}

\newcommand{\REM}[1]{}

\newcommand{\RR}{\mathcal{R}}
\newcommand{\HH}{\mathcal{H}}
\newcommand{\D}{\mathcal{D}}
\newcommand{\numcopies}{\sum_{h\in\HH} q^{-}(h)}

\newenvironment{appendix-lemma}[1]{\vspace{0.1in}\noindent{\bf Lemma~#1~} \em }{\vspace{0.1in}}
\newenvironment{appendix-theorem}[1]{\vspace{0.1in}\noindent{\bf Theorem~#1~} \em }{\vspace{0.1in}}

\newtheorem{invariant}{Observation}
\newtheorem{observation}{Observation}
\title{Popular Matchings with Lower Quotas}
\author[1]{Meghana Nasre}
\author[2] {Prajakta Nimbhorkar}
\affil[1]{Indian Institute of Technology, Madras, India \tt{(meghana@cse.iitm.ac.in)}}
\affil[2]{Chennai Mathematical Institute, India \tt{(prajakta@cmi.ac.in)}}
\begin{document}
\maketitle
\begin{abstract}
We consider the well-studied Hospital Residents (HR) problem in the presence
of lower quotas (LQ). The input instance consists of a bipartite graph $G = (\RR \cup \HH, E)$
where $\RR$ and $\HH$ denote sets of residents and hospitals respectively. 
Every vertex has a preference list that imposes a strict ordering on its neighbors.
In addition, each hospital $h$ has an associated
upper-quota $q^+(h)$ and lower-quota $q^-(h)$. 
A matching $M$ in $G$ is an assignment of residents to hospitals, and $M$ is said to be feasible
if every resident is assigned to at most one hospital and a hospital $h$ is assigned at least $q^-(h)$ 
and at most $q^+(h)$ residents.

Stability is a de-facto notion of optimality in a model where both sets of vertices have preferences.
A matching is stable if no unassigned pair has an incentive to deviate from it.
It is well-known that an instance of the HRLQ problem need not admit a feasible stable matching.
In this paper, we consider the notion of popularity for the HRLQ problem. A matching $M$ is popular
if no other matching $M'$ gets more votes than $M$ when vertices vote between $M$ and $M'$.
When there are no lower quotas, there always exists a stable matching and it is known that every
stable matching is popular. 

We show that in an HRLQ instance, although a feasible stable matching need
not exist, there is always a matching that is popular in the set of feasible matchings. 
We give an efficient algorithm to compute a maximum
cardinality matching that is popular amongst all the feasible matchings in an HRLQ instance.
%Our algorithm is a variant of the Gale-Shapley algorithm and it works by reducing the given HRLQ
%instance to an HR instance that has no lower-quotas. We show that every stable matching in 
%the new instance can be mapped to a maximum cardinality matching that is popular amongst all 
%the feasible matching in the given instance. 
\end{abstract}

\section {Introduction}
In this paper we consider the Hospital Residents problem in the presence of Lower Quotas (HRLQ). The
input to our problem is a bipartite graph $G = (\RR \cup \HH, E)$ where $\RR$ denotes the
set of residents, and $\HH$ denotes the set of hospitals.
Every resident
as well as hospital has a non-empty preference ordering over a subset of elements of the other set.
Every hospital $h \in \HH$
has a non-zero upper-quota $q^+(h)$ denoting the maximum number of residents
that can be assigned to $h$. In addition, every hospital $h$ also has a non-negative lower-quota $q^-(h)$ denoting
the minimum number of residents that have to be assigned to $h$. The goal is to assign residents to hospitals such that the upper and lower quotas of all the hospitals
are respected (that is, it is feasible) as well as the assignment is {\it optimal} with respect to the preferences of the participants. 
\begin{definition}\label{def:feasible}
A feasible matching $M$ in $G=(\RR\cup\HH,E)$ is a subset of $E$ such that $|M(r)|\leq 1$ for
each $r\in \RR$ and $q^-(h)\leq |M(h)|\leq q^+(h)$ for each $h\in \HH$, where $M(v)$ is the 
set of neighbors of $v$ in $M$.
\end{definition} 
{\it Stability} is a de-facto notion of optimality in settings where both sides have preferences. 
A matching $M$ (not necessarily feasible) is said to be stable if there is no {\it blocking pair} with respect to $M$. A resident-hospital pair $(r,h)$ blocks $M$ if $r$ is unmatched in $M$ or prefers $h$ over $M(r)$,
and either $|M(h)|<q^+(h)$ or $h$ prefers $r$ over at least one resident in $M(h)$.
%To define stability, we introduce some notation.
%Let $M(r)$ denote the hospital to which $r$ is matched in $M$. Similarly, let  $M(h)$ denote the set of residents that are matched to $h$ in $M$. A hospital $h$ is {\it under-subscribed} in $M$ if $|M(h)| < q^+(h)$; hospital $h$ is {\it deficient} if $|M(h)| < q^{-}(h)$.
%\begin{definition}
%\label{def:stab}
%A pair $(r, h) \in E \setminus M$ blocks $M$ if either $r$ is unmatched in $M$ or $r$ prefers $h$ over $M(r)$
%and either $h$ is under-subscribed in $M$ or $h$ prefers $r$ over at least one resident in $M(h)$. A matching $M$ is called {\it stable} if there
%does not exist any blocking pair w.r.t. $M$.
%\end{definition}

There are simple instances of the HRLQ problem where there is no feasible matching that is stable. We give an example here: Let $\RR=\{r\}, \HH=\{h_1,h_2\}$, $q^+(h_1)=q^+(h_2)=1$, $q^-(h_1)=0$, and $q^-(h_2)=1$. Let preference list of $r$ be $\langle h_1,h_2\rangle$. That is, $r$ prefers $h_1$ over $h_2$.
The only stable matching here is $M_1 = \{(r, h_1)\}$ which is not feasible as $|M_1(h_2)|<q^-(h_2)$. On the other hand, the only feasible matching $M_2=\{(r,h_2)\}$ is not
stable as $(r,h_1)$ is a blocking pair with respect to $M_2$.
This raises the question, given an HRLQ instance $G$, does $G$ admit a feasible stable matching? 
This can be answered by constructing an HR instance $G^+$ by disregarding the lower quotas of all hospitals in $G$.
It is well-known that the Gale-Shapley algorithm~\cite{GS62} computes a stable matching $M$ in $G^+$. Furthermore,
from the ``Rural Hospitals Theorem" it is known that, in every stable matching of $G^+$, %matches the same set of residents
each hospital is matched to the same capacity~\cite{GS85, Roth86}. Thus $G$ admits a stable feasible matching if and only if $M$ is 
feasible for $G$.
%MN: I have reworded this because I feel that the Rural Hospitals Theorem applies only to HR instances. So we should distinguish that we are
% using RHT for G+.
%MN: Prajakta's text is retained as it is below.
\REM{
It follows from the 
``Rural Hospitals Theorem'' that either all the stable matchings in $G$ are feasible 
or all are infeasible. This is because the theorem states that each stable matching matches the same set of residents
and each hospital is matched to the same capacity in every stable matching \cite{GS85, Roth86}.
This makes it easy to check whether a given HRLQ instance admits a feasible stable matching. It can be done in polynomial time by simply disregarding the lower quotas of all the hospitals and computing a stable
matching in the resulting HR instance $G^+$ by Gale-Shapley algorithm \cite{GS62}. 
}

The HRLQ problem is motivated by practical scenarios like assigning medical interns (residents) to hospitals.
While matching residents to hospitals,
%it is necessary for smooth functioning of a hospital that its minimum requirement on the number of residents is satisfied. 
rural
hospitals often face the problem of being understaffed with residents, for example the  National Resident Matching Program in the US \cite{BrandlK16, Roth84, Roth86}.
In such real-world applications declaring that there is no feasible stable matching is simply not a solution. On the other hand,
any feasible matching that disregards the preference lists completely is socially unacceptable.
We address this issue by relaxing the requirement of stability by an alternative notion of namely {\it popularity}.
Our output matching $M$ has two desirable criteria -- firstly, it is a feasible matching in the instance, assuming one such exists, and hence no hospital
remains understaffed. Secondly, the matching respects preferences of the participants, in particular,  no majority of participants wish to deviate to another 
feasible matching in the instance.

%MN: reworded a bit again -- retained Prajakta's text as it is.
\REM{
From a practical perspective, the HRLQ problem arises naturally. While matching residents to hospitals,
it is necessary for smooth functioning of a hospital that its minimum requirement on the number of residents is satisfied. For instance, in the National Residents Matching Program in the USA, rural
hospitals often face the problem of being understaffed with residents \cite{BK16, Roth84, Roth86}.
In such a setting, the option of finding a maximum cardinality matching disregarding the preference lists is socially unacceptable. Our algorithm achieves two goals here. Firstly, it outputs a matching that respects
upper as well as lower quotas of all the hospitals as long as such a matching exists, and hence no
hospital remains understaffed. Secondly, the notion of popularity used by our algorithm is a relaxation
of stability. Hence maximum cardinality popular matchings are typically significantly larger than stable matchings and hence
leave a smaller number of residents unmatched (see e.g. \cite{??}).
}
\noindent {\bf Our contribution: }
We consider the notion of {\it popularity} for the HRLQ problem. Popularity is a relaxation of stability
and can be interpreted as {\it overall stability}. We define it formally in Section~\ref{sec:pop-def}. In this work, we present an efficient algorithm for the following two problems in an HRLQ instance.
\begin{enumerate}
\item Computing a maximum cardinality matching popular in the set of feasible matchings. We give an $O(|\RR|(|E|+|\HH|))$ time algorithm
for this problem.
\item Computing a popular matching amongst maximum cardinality feasible matchings. We give an $O(|\RR|^2(|E|+|\HH|))$ time algorithm for this problem.
\end{enumerate}
%Our algorithms run in time $O(|\RR|(|E|+|\HH|))$ and $O(|\RR|^2(|E|+|\HH|))$ respectively.
Our algorithms are based on ideas introduced in earlier works on stable marriage (SM) and HR problems\cite{Kavitha14,BrandlK16,NasreR16}. However, in SM and HR problem, a popular matching
is guaranteed to exist because a stable matching always exists and it is also popular. 
On the other hand, in the HRLQ setting even a stable matching may not exist. Yet, we prove that
a feasible matching that is popular amongst all feasible matchings always exists and is efficiently computable. We believe that this is not only  surprising
but also a useful result  in practical scenarios. Moreover, our notion of popularity subsumes the notions 
proposed in \cite{BrandlK16} and \cite{NasreR16} and is more general than both. In \cite{BrandlK16}, popularity is proved using linear programming, but our proofs for popularity are combinatorial.
 
\noindent {\bf Overview of the algorithm:} Our algorithms are reductions, that is, given an HRLQ instance $G$,
both our algorithms construct instances $G'$ and $G''$ of the HR problem such that there is a natural way to map a stable matching in $G'$ (respectively, $G''$) to a feasible matching
in $G$. Moreover, any stable matching in $G'$ ($G''$) gets mapped to a maximum cardinality matching that is popular amongst all the feasible matchings in $G$ (respectively, a matching that is popular amongst all maximum cardinality matchings in $G$). 

\noindent {\bf Organization of the paper:} We define the notion of popularity in Section~\ref{sec:pop-def}.
The reduction for computing a maximum cardinality popular matching amongst feasible matchings is given in Section~\ref{sec:red} and its correctness is proved in Section~\ref{sec:pop}. Finally, Section~\ref{sec:popmax} describes the reduction for computing a matching that is popular amongst maximum 
cardinality feasible matchings and its correctness.

\noindent {\bf Related work:} 
The notion of popularity was first proposed by G\"ardenfors \cite{G75}
in the stable-marriage (SM) setting, where each vertex has capacity $1$, and have been well-studied
since then \cite{BIM10, HK13, Kavitha14, HYKY15, CsehK16, Kavitha16}.
%It is well-known that, in the stable marriage (SM) and HR settings, a stable matching is popular (under a well-defined notion of popularity for the HR problem).
%Thus, every instance of the 
%HR problem admits a popular matching. 
A linear-time algorithm
to compute a maximum cardinality popular matching in an HR instance is given in ~\cite{BrandlK16} and \cite{NasreR16} with different notions of popularity. Furthermore,
for the SM and HR problem, it is known that a matching that is {\it popular amongst the maximum cardinality
matchings} exists and can be computed in $O(mn)$ time \cite{Kavitha14, NasreR16}.
The  reductions in our paper are inspired by the work of \cite{BrandlK16,CsehK16,Kavitha14,NasreR16}.
In all these earlier works, the main idea is to execute Gale and Shapley algorithm
on the HR instance and then allow unmatched residents to propose with {\it increased priority}~\cite{Kavitha14} certain number of times.
As mentioned in~\cite{Kavitha14}, this idea was first proposed in ~\cite{Kiraly11}.
%For the HRLQ problem, Hamada~et~al.~\cite{Hamada16} address 
%the problem of computing a matching with minimum number of blocking pairs (MinBP) and that of computing
%a matching with the minimum number of residents participating in a blocking pair (MinBR).
%In \cite{Hamada16}, the authors proved NP-Hardness and gave approximation algorithms for
%both these problems, but these algorithms work only when all the hospitals with non-zero lower-quota have complete preference lists.
%($|\HH| + |\RR|$)-approximation algorithm
%for the Min-BP problem and a $\sqrt{|\RR|}$-approximation algorithm for the Min-BR problem under a restricted model, called the {\it CL-restriction}, where all the hospitals with non-zero lower quota have complete preference lists.
%To the best of our knowledge, there are no known algorithms that compute a feasible matching that respects preferences in a general HRLQ problem. 
%The HR problem with lower quotas was considered by Bir\'{o}~et~al.~\cite{Biro10}. The model
%used in their paper is a variant of the one in our paper -- in \cite{Biro10}, hospitals can be closed, however such a provision is not present in our model. In the model of \cite{Biro10}, it is NP-complete
%to determine the existence of a stable matching, which is not the case with our model.
The HRLQ problem has been recently considered in \cite{Biro10} and \cite{Hamada16} in different 
settings.

\section{Notion of Popularity}\label{sec:pop-def}
The notion of popularity uses votes from vertices to compare two matchings.
For $r \in \RR$, and any matching $M$ in $G$, if $r$ is unmatched in $M$ then,
$M(r) = \bot$.
A vertex prefers any of its neighbours over $\bot$.
For a vertex $u \in \RR \cup \HH$, let $x, y \in N(u) \cup \{\bot\}$,
where $N(u)$ denotes the neighbours of $u$ in $G$. We define $vote_u(x,y)=1$ if $u$ prefers $x$ 
over $y$, $-1$ if $u$ prefers $y$ over $x$ and $0$ if $x=y$.
%\begin{eqnarray*}
%vote_u(x, y) &=& +1 \ \ \ \ \mbox { if $u$ prefers  $x$ over $y$} \\ 
 %            &=& -1 \ \ \ \ \mbox { if $u$ prefers $y$ over $x$} \\
  %           &=& \ \  0 \ \ \ \ \mbox { if $x$ = $y$ }
%\end{eqnarray*}
%We now describe how the vote of a vertex between two neighbours can be used to
%translate into a vote between two matchings.
Given two  matchings $M_1$ and $M_2$ in the instance, 
%every participant (a resident
%as well as hospital) is given one or more votes, denoted by $vote_u(M_1, M_2)$ which allows them to declare which 
%amongst $M_1$ and $M_2$ is preferrable to the participant.
%Each resident is given exactly one vote and a hospital $h$ is assigned $q^+(h)$ votes.
for a resident $r \in \RR$, we define $vote_r(M_1, M_2)  = vote_r(M_1(r), M_2(r))$.
%prefers $M_1$ over $M_2$ and therefore
%votes for $M_1$ over $M_2$ if 
%$r$ is matched in $M_1$ and unmatched
%in $M_2$ or $r$ is matched in both $M_1$ and $M_2$ and $r$ prefers $M_1(r)$ over $M_2(r)$.
%Thus 

\noindent {\bf Voting for a hospital:} A hospital $h$ is assigned $q^+(h)$-many
votes to compare two matchings  $M_1$ and $M_2$; this can be viewed as one vote
per position of the hospital. If a position is not filled in a matching, we put a $\bot$ there, so that
$|M_1(h)|=|M_2(h)|=q^+(h)$.
%Let $same(h, M_1, M_2) = |M_1(h) \cap M_2(h)|$ denote the number of residents that
%are assigned to $h$ in both $M_1$ and $M_2$. Furthermore, let $free(h, M_1, M_2) = q^+(h) - \max\{|M_1(h)|, |M_2(h)|\}$
%denote the positions of $h$ which are unfilled in both $M_1$ and $M_2$.
In our voting scheme, the hospital $h$ is indifferent between $M_1$ and $M_2$ as far as its
$|M_1(h)\cap M_2(h)|$ positions are concerned.

%In order to use rest of its votes, 
To compare between the two sets of residents $M_1(h) \setminus M_2(h)$ and $M_2(h) \setminus M_1(h)$, a hospital can decide any pairing of the elements of these two sets. We denote this
correspondence by 
%We define a function 
${\bf corr_h}$. Under this correspondence, for a resident $r\in M_1(h)\setminus M_2(h)$, $\textrm{{\bf corr}}_h(r,M_1,M_2)$ is the resident in $M_2(h)\setminus M_1(h)$ corresponding to $r$. We define
% which allows a hospital to declare the correspondence between
%residents of $M_1(h) \setminus M_2(h)$ and $M_2(h) \setminus M_1(h)$. 
%These correspondences
%will be used when $h$ compares   $M_1$ and $M_2$.
%The function ${\bf corr}$ takes as its input a resident $r \in M_1(h) \oplus M_2(h)$, 
%the two matchings $M_1$ and $M_2$ and outputs a resident in $M_1(h) \oplus M_2(h) $ or $\bot$. 
%Note that every resident has a unique corresponding resident.
%\begin{eqnarray*}
%{\bf corr_h} (r, M_1, M_2) \rightarrow M_1(h) \oplus M_2(h) \cup \{\bot\}
%\end{eqnarray*}

\begin{eqnarray*}
vote_h(M_1, M_2) & = &\sum_{r\in M_1(h)\setminus M_2(h)} vote_h(r,\textrm{{\bf corr}}_h(r, M_1,M_2))
%(|M_1(h)| - |M_2(h)|) +
%\sum_{\substack {r \in M_1(h) \setminus M_2(h) \\  { \&\& }\\ {\bf corr}(r, M_1, M_2) \neq \bot}} vote_h(r, {\bf corr}(r, M_1, M_2))
\end{eqnarray*} 
A hospital $h$ prefers $M_1$ over $M_2$ if $vote_h(M_1, M_2) > 0$.
We can now define popularity.
\begin{definition}
\label{def:pop}
A matching $M_1$ is more popular than $M_2$ (denoted as $M_1 \succ M_2$) if 
$\sum_{v\in \RR\cup \HH}vote_v(M_1,M_2)>\sum_{v\in \RR\cup \HH}vote_v(M_2,M_1)$.
%the number
%of votes that $M$ gets as compared to $M'$ is greater than the number of votes that $M'$ gets as compared to $M$.
A matching $M$ is popular if there is no matching $M'$ such that $M' \succ M$.
\end{definition}

There are several ways for a hospital to define the ${\bf corr}$ function. For example, 
a hospital $h$ may decide to order and compare the two sets in the decreasing order of preferences (as in \cite{NasreR16}
%that the most preferred resident in $M_1(h) \setminus M_2(h)$ will be compared
%to the most preferred resident in $M_2(h) \setminus M_1(h)$ and so on. Another hospital may decide
%to pair up residents in $M_1(h) \setminus M_2(h)$ and $M_2(h) \setminus M_1(h)$ 
or in the {\it most adversarial order} (as in \cite{BrandlK16}).
That is, the order due to which $h$ gives the least votes to $M_1$ when comparing it with $M_2$.
We believe that our definition offers flexibility to hospitals to compare residents in $M_1(h) \setminus M_2(h)$ 
and $M_2(h) \setminus M_1(h)$ according to their custom designed criteria. 
%irrespective of how hospitals compare residents in the symmetric difference.

\noindent {\bf Decomposing $M \oplus M'$:} 
In the one-to-one setting, where $M\oplus M'$ for any two matchings $M$ and $M'$ is a collection
of vertex-disjoint paths and cycles. Our setting is many-to-one and hence $M\oplus M'$ has a more
complex structure. 
Here, we recall a simple algorithm to decompose edges of $M \oplus M'$ into (possibly non-simple) alternating paths and cycles from \cite{NasreR16}. %This decomposition will be used in the subsequent sections and all of our proofs.
Consider the graph  $\tilde{G} = (\RR \cup \HH, M \oplus M')$, for any two feasible matchings of the HRLQ instance.
We note that the degree of every resident in $\tilde{G}$ is at most $2$ and
the degree of every hospital in $\tilde{G}$ is at most $2 \cdot q^+(h)$.
Consider any connected component $\mathcal{C}$ of $\tilde{G}$
and let $e\in M$ be any edge in $\mathcal{C}$. We 
show how to construct a unique maximal $M$-alternating path or cycle $\rho$ containing
$e$: Start with $\rho=\langle e \rangle$. Use the following inductive procedure.
\begin{enumerate}
\item Let $r \in \RR$ be one end-point of $\rho$, and let $(r, M(r)) \in \rho$. % (resp. $(r, M'(r)) \in \rho$).
We grow $\rho$ by adding the edge $(r, M'(r))$. %(resp. $(r, M(r))$) if such an edge exists.
Similarly if $(r,M'(r))\in \rho$, add $(r,M(r))$ to $\rho$. 
%If $M(r)$ or $M'(r)$ (both denoted as $h$) has a degree greater than two in $\tilde{G}$
%then we add a duplicate copy of $h$ on this path, and simply call it $h$.

\item Let $h \in \HH$ be an end-point of $\rho$, and let the last edge $(r, h)$ on $\rho$ be in $M \setminus M'$.
We extend $\rho$ by adding $\cor_h(r,M,M')$ if is not equal to $\bot$.
A similar step is performed if the last edge on $\rho$ is $(r, h) \in M' \setminus M$.
\item We stop the procedure when we complete a cycle (ensuring that the two adjacent residents of a hospital are ${\bf corr}$ for each other according to the hospital), or the path can no longer be extended.
Otherwise we go to Step~1 or Step~2 as applicable and repeat.
\end{enumerate}

\noindent {\bf Labels on edges:} While comparing a matching $M_1$ with another matching $M_2$, the voting scheme induces a label on edges of $M_2$ with respect to $M_1$. Let $(r,h)\in M_2$. The label of $(r,h)$ is 
$(a,b)$ where 
$a=vote_r(M_1(r),M_2(r))$ and $b=vote_h(\cor_h(r,M_2,M_1),r)$. Thus $a,b\in \{-1,1\}$.

\section{Reduction to HR problem }\label{sec:red}
In this section we present our reduction from an HRLQ instance $G = (\RR \cup \HH, E)$ to an HR instance $G' = (\RR' \cup \HH', E')$.
To compute a largest size feasible matching that is popular amongst all feasible matchings, we compute a stable matching
$M'$ in $G'$. We  show that there is a natural map from any stable matching $M'$ in $G'$ to a
feasible matching $M$ in $G$.

Before we describe the reduction in detail, we provide some intuition. Our reduction simulates the following algorithm:
Execute the hospital-proposing Gale-Shapley algorithm on $G$ by disregarding lower quotas of all hospitals. Let $M_0$ be a matching obtained. If some hospitals are {\it under-subscribed}\footnote{We say that a hospital is under-subscribed in a matching $M$ if $|M(h)|<q^+(h)$ and is deficient if $|M(h)|<q^-(h)$} in $M_0$,
they apply with increased priority to residents (in order of preference) and a new matching $M_1$ is obtained. If there are {\it deficient hospitals} in $M_1$, they again apply with an even higher priority. This process is repeated until there is no deficient hospital. This is achieved by reducing $G$ to an HR instance $G'$
described below:

%Given an HRLQ instance $G=(\RR\cup\HH, E)$, we give a reduction from $G$ to an HR instance $G'=(\RR'\cup \HH',E')$
%such that a stable matching in $G'$ is a maximum cardinality popular matching in $G$.
%The reduction is as follows:

%Let us call $h \in \HH$ as a {\it lower quota hospital} if $q^-(h) > 0$, otherwise $h$ is called a {\it non-lower quota hospital}.
%Thus, $\HH_{\ell q} = \{ h \in \HH \mid q^-(h) > 0\}$ and $\HH_{u q } = \HH \setminus \HH_{\ell q}$.
We first describe the vertices in $G'$.

\noindent {\bf \underline {The set $\HH'$:}} 
For each hospital $h \in \HH$ we have $\ell$ copies $h^0, \ldots, h^{\ell-1}$  of $h$ in $\HH'$.  
Here $\ell =  2+ \numcopies$. We need to define the capacities\footnote{We use the term capacity for the hospitals in an HR instance whereas the term quota for hospitals in an HRLQ instance.
} of all hospitals $h \in \HH'$ (recall $G'$ is an HR instance, so we do not have lower quotas for $h \in \HH'$).
The hospitals in $\HH'$ and their capacities are as described below:
\begin{eqnarray*}
\HH' & = & \{h^0,\ldots, h^{\ell-1}\mid h\in \HH\} \\
\textrm{Capacities of $h \in \HH'$: }q^+(h^s) & = & q^+(h), \quad s\in \{0,1\}\\
q^+(h^s) & = & q^-(h), \quad s\in \{2,\ldots,\ell-1\}
\end{eqnarray*}
We call hospital $h^s \in \HH'$ a level-$s$ copy of $h$.
Note that if $h \in \HH$ has zero lower-quota,  then  $h^1,\ldots, h^{\ell -1}$  have zero
capacity in $\HH'$.
For a hospital $h \in \HH$, we denote by $q_h$ the sum of the capacities of all level copies of $h$ in $G'$. 
The following observation is immediate.
\begin{observation}
\label{inv:sum-cap-H}
For a hospital $h \in \HH$, the sum of capacities of all level copies of $h$ in $G'$ is $q_h = 2 \cdot q^+(h) + (\ell - 2) \cdot q^-(h)$.
\end{observation}

\vspace{0.1in}

\noindent {\bf \underline{The set $\RR'$:}} The set of residents $\RR'$ consists of the set $\RR$ along
with a set of dummy residents $\mathcal{D}_h$ corresponding to every hospital $h \in \HH$.
The set $\RR'$ and $\mathcal{D}_h$ are as defined below:
\begin{eqnarray*}
\RR' & = & \RR\cup \left(\bigcup_{h\in \HH} \mathcal{D}_h\right)\textrm{ where }
 \ \ \ \ \ \ \ \ \ \ \ \quad \mathcal{D}_h  =  \bigcup_{s\in\{0,\ldots,\ell-2\}}\mathcal{D}^s_h \ \ \  \ \ \ \ \ \forall h \in \HH\\
\vspace{0.1in}
\textrm{Here }\mathcal{D}^s_h& = & \{d^s_{h,1},\ldots,d^s_{h,q^+(h)}\}, \quad s\in\{0,1\}\\
\textrm{and }\mathcal{D}^s_h&=&\{d^s_{h,1},\ldots,d^s_{h,q^-(h)}\},\quad {s\in \{2,\ldots,\ell-2\}}
\end{eqnarray*}
We refer to $\mathcal{D}_h$ as {\it dummy residents corresponding to $h$} and $\mathcal{D}^s_h$ as {\it level-$s$
dummy residents corresponding to $h$}.
For $h \in \HH$, if $q^-(h) = 0$, then $\mathcal{D}^s_h=\emptyset$ for each $s\in\{2,\ldots,\ell-1\}$.% if $q^-(h)=0$.

\noindent The following observation captures the number of dummy residents corresponding to every hospital $h \in \HH$.
\begin{observation}
\label{inv:sum-dummy-H}
For a hospital $h \in \HH$, the total number of dummy residents corresponding $h$ in $\RR'$ is $|\D_h| = 2 \cdot q^+(h) + (\ell - 3) \cdot q^-(h)$.
\end{observation}

\vspace{0.1in}
\noindent {\bf \underline{Preference lists:}} We denote by $\langle list_r\rangle$ and $\langle list_h \rangle$
the preference lists of $r$ and $h$ in $G$ respectively. Furthermore, $\langle \mathcal{D}^s_h\rangle$ denotes the strict list consisting of
elements of $\mathcal{D}^s_h$ in increasing order of indices. Finally, $\circ$ denotes the concatenation of two lists.  We now describe the preferences of hospitals and residents in $G'$.

\vspace{0.1in}
\noindent {\em Hospitals' preference lists:} Consider a hospital $h^s \in \HH'$ for $s \in \{1, \ldots, \ell-2\}$ and let $q$ denote the capacity of $h^s$.
The preference list of $h^s$ is of the form: $q$-dummy residents of level-$(s-1)$, followed by preference list of $h$ in $G$, followed
by $q$ dummy residents  of level-$s$. For $h^0$, the preference list is the preference list of $h$ in $G$ followed by capacity many dummy residents of level-0.
Finally, for $h^{\ell-1}$, there are dummy residents of level-$(\ell-2)$ followed by preference list of $h$ in $G$.
For $h^s\in \HH'$,
\[
\begin{array}{lcl}
s=0& :& \langle list_h\rangle \circ \langle \mathcal{D}^0_h \rangle\\
s=1 &:& \langle \mathcal{D}^{0}_h\rangle \circ \langle list_h\rangle \circ \langle \mathcal{D}^1_h\rangle \\
s=2 &:&  \langle d^1_{h,k},\ldots, d^1_{h,q^+(h)}\rangle \circ \langle list_h\rangle \circ \langle \mathcal{D}^2_h\rangle, \qquad k=q^+(h)-q^-(h)+1\\
s\in \{3,4,\ldots, \ell-2\} &:& \langle \mathcal{D}^{s-1}_h\rangle \circ \langle list_h\rangle \circ \langle \mathcal{D}^s_h\rangle\\
s=\ell-1 &:& \langle \mathcal{D}^{(\ell-2)}_h\rangle\circ\langle list_h\rangle
\end{array}
\]
\noindent {\em Residents' preference lists:}
\[
\begin{array}{lcl}
\textrm{For }r\in \RR &:& \langle list_r \rangle^{\ell-1} \circ \langle list_r\rangle^{\ell-2}\circ\ldots\circ\langle list_r\rangle^0\\
\textrm{For }h\in \HH, \ \ \ d^s_{h,i}\in \mathcal{D}_h &:& \\
s = 0 & : & \langle h^0, h^1 \rangle \\
s = 1, \ \  i\in \{1,\ldots,q^+(h)-q^-(h)\} &: & \langle h^1 \rangle\\
s = 1, \ \ i\in \{q^+(h)-q^-(h)+1,\ldots,q^+(h)\} & : & \langle h^1, h^2 \rangle\\
s \in \{2, \ldots, \ell -2 \} &:& \langle h^s, h^{s+1} \rangle
%\langle h^1\rangle & & \textrm{if }s=1, i\in \{1,\ldots,q^+(h)-q^-(h)\}\\
%\langle h^1, h^2\rangle & & \textrm{if } s=1, i\in \{q^+(h)-q^-(h)+1,\ldots,q^+(h)\}\\
%\langle h^s, h^{s+1} \rangle & & \textrm{if }s\in \{0,2,\ldots,\ell-2\}
\end{array}
\]

\subsection{Properties of the stable matching $M'$ in $G'$}
With respect to a stable matching $M'$ in $G'$ we introduce the following definitions. \begin{definition}
{\bf Level-$s$ resident:} A non-dummy resident $r \in \RR'$ is said to be at level-$s$ in $M'$ if
$r$ is matched to a level-$s$ hospital in $M'$. Let $\RR'_s$ denote the set of level-$s$ residents. 
\end{definition}

\begin{definition}
{\bf Active hospital:}
A hospital $h^s$ is said to be {\it active} in $M'$ if $M'(h^s)$
contains at least one non-dummy resident. Otherwise, (when all positions of $h^s$ are matched to dummy residents), $h^s$ is said to be {\it inactive}.
\end{definition}
In the following lemma, we state some invariants for any stable matching $M'$ in $G'$. These invariants allow us to
define a natural map from $M'$ to a matching $M$ in $G$, and to show that $M$ is feasible as well as popular among feasible matchings.

\begin{lemma}\label{lem:SM-invariants}
\begin{enumerate}
The following hold for any stable matching $M'$ in $G'$:
\item\label{inv:one} For any $h \in \HH$, $M'$ matches at most $q^+(h)$ non-dummy residents across all its 
level copies in $G'$.
\item \label{inv:two}The matching $M'$ in $G'$ leaves only the level-$(\ell -1)$ copy of any hospital (if it exists) under-subscribed. 
\item Let $h^s \in \HH'$ be active in $M'$.  Then,
\begin{enumerate}
\item\label{inv:four-one} At least one position of $h^{s-1}$ is matched in $M'$ to a dummy resident at level-($s-1$).
\item\label{inv:four-two} For $0 \le j \le s-2$, $h^j$ is inactive in $M'$ and all positions of $h^{j}$ are matched to dummy residents of level-($j$).
\item\label{inv:four-three} For $s+2 \le j \le \ell-1$, $h^j$ is inactive in $M'$ and all  positions of $h^{j}$ are matched to  dummy residents of level-($j-1$).
\end{enumerate} 
\item\label{inv:three} For any $h \in \HH$, at most two consecutive level copies $h^s$ and $h^{s+1}$ are active in $M'$.
\item\label{inv:steep} A level-$s$ resident $r$ in $M'$ does not have any hospital $h$ in its preference list which is active at level-$(s+2)$ or more in $M'$.
%\textcolor{blue}{Proof: Otherwise the edge $(r, h^{s+1})$ forms a blocking pair with respect to $M'$.}
\end{enumerate}
\end{lemma}

\begin{proof}
\begin{itemize}
\item {\em Proof of \ref{inv:one}}: Consider the set of dummy residents corresponding to a hospital 
$h\in \HH$ i.e. $\bigcup_{s=0}^{\ell-2}\mathcal{D}^s_h$. With the exception of $h^2$, for any $h^s$, $\D^{(s-1)}_h$ are
the most preferred $q^+(h^s)$ dummy residents of $h^s$. Thus these dummy residents can never
remain unmatched in $M'$. The only dummy residents that are not the first choice of any hospital and hence can remain unmatched are the subset of $\D^1_h$ consisting of the first $q^+(h)-q^-(h)$ dummy residents from $\D^1_h$. This is because, by construction of $G'$, only the last $q^-(h)$ dummy residents from $\D^1_h$
are present in the preference list of $h^2$ as its top $q^+(h^2)$ top-choices. 

Thus the total number of dummy residents for $h$ is given by $|\D_h|=2\cdot q^+(h)+(\ell-3)\cdot q^-(h)$.
Total capacity of all the copies of $h$ in $G'$ is $q_h=2\cdot q^+(h)+(\ell-2)\cdot q^-(h)$. Number of dummy residents in $\D_h$ that can remain unmatched in any stable matching $M'$ in $G'$ is at most $q^+(h)-q^-(h)$.
Thus the number of true residents matched to $h$ in $M'$ is at most $q_h-|\D_h|+q^+(h)-q^-(h)$
which is $q^+(h)$.

\item {\em Proof of \ref{inv:two}:} Consider a hospital $h\in\HH$. For each copy $h^s$ of $h$ in $\HH'$, 
where $s<\ell-1$, the dummy residents of level $s$ have $h^s$ as their first choice. Further, their
number is same as $q^+(h^s)$. Thus $h^s$ can not remain undersubscribed in any stable matching $M'$ of $G'$, otherwise these dummy residents will form a blocking pair with $h^s$.

\item {\em Proof of \ref{inv:four-one}:} For the sake of contradiction, assume that $h^{s-1}$ is not 
matched to any level-$(s-1)$ dummy resident and still $h^s$ is matched to a non-dummy resident.
As there are exactly $q^+(h^s)$ many level-$(s-1)$ dummy residents in the preference list of $h^s$,
and each level-$(s-1)$ dummy resident has only $h^{s-1}$ and $h^s$ in its preference list,
this means that there is a level-$(s-1)$ dummy resident $d$ unmatched in $M'$. 
But $h^s$ prefers any level-$(s-1)$ dummy resident over any non-dummy resident. Thus $(d,h^s)$ forms
a blocking pair with respect to $M'$, contradicting the stability of $M'$.

{\em Proof of \ref{inv:four-two}:} If $h^s$ is active and $h^j$ is matched to a non-dummy resident $r$ for some $0\leq j\leq s-2$, then 
$(r,h^{(s-1)})$ is a blocking pair with respect to M'. This is because, as proved above, $h^{(s-1)}$
must be matched to at least one resident in $\D^{(s-1)}_h$, and $h^{(s-1)}$ prefers any non-dummy
resident over any dummy resident in $\D^{(s-1)}_h$.

{\em Proof of \ref{inv:four-three}:} If $h^s$ is active then $h^j$ can not be active for $s+2\leq j\leq \ell-1$ else $h^{(s+1)}$ must be
matched to a resident from $D^{(s+1)}_h$ as proved above, and then each non-dummy resident $r$
in $M'(h^s)$ forms a blocking pair with $h^j$ contradicting the stability of $M'$. But if $h^s$ is active,
then $h^j$ can not be matched to a dummy resident from $\D^j_h$ either, otherwise a resident in
$M'(h^s)$ forms a blocking pair with $h^j$. The later is true because any resident in $list_h$ prefers
$h^j$ over $h^s$ for $j>s$ and $h^j$ prefers any resident in $list_h$ to any dummy resident in $\D^j_h$.
Hence $h^j$ must be matched to only dummy residents in $\D^{(j-1)}_h$.  
%\textcolor{blue}{I am not sure if we need other two parts of this invariant. I'll write the proofs later if
%they are being used.}
\item {\em Proof of \ref{inv:three}:} Assume the contrary. Thus let $h$ be a hospital such that there are two levels $i$ and $j$, $j<i-1$,
where $h^i$ and $h^j$ are matched to one or more non-dummy residents. Further, assume that $h^i$ is matched to $r_i$ and $h^j$ be matched to $r_j$. Then, by Invariant $4$ above, $h^{i-1}$ must be
matched to at least one $(i-1)$-level dummy resident. But, by the structure of preference lists,
$h^{i-1}$ prefers a non-dummy resident, and hence $r_j$, over any $(i-1)$-level dummy resident.
Also, $r^j$ prefers $h^{i-1}$ over $h^j$ since $j<i-1$. Thus $(r_j, h^{i-1})$ forms a blocking pair
in $G'$ with respect to $M'$, contradicting the stability of $M'$.
\item {\em Proof of \ref{inv:steep}:} Let there be an edge $(r,h^t)$ in $G'$ such that $r$ is a level-$s$ 
resident and $t\geq s+2$ and $h^t$ is active in $M'$. Then, by Invariant \ref{inv:four-one}, $h^{s+1}$
has at least one level-$(s+1)$ dummy resident in $M'(h^{s+1})$. As $r$ has edge to $h^t$, $r$ also
has an edge to $h^{s+1}$ by construction of $G'$. Also, again by construction of $G'$, $r$ prefers
any level-$(s+1)$ hospital over any level-$s$ hospital and $h^{s+1}$ prefers any non-dummy 
resident in its preference list over any level-$(s+1)$ dummy resident. Thus $(r,h^{s+1})$ forms a 
blocking pair with respect to $M'$ in $G'$, contradicting its stability.
\end{itemize}
\end{proof}

%\subsection{Map of $M'$}
%In this section we define the map of $M'$ to a matching $M$ in $G$. We also prove that if $G$ admits a feasible matching
%then $map(M')$ is feasible for $G$.

% MN: (April 4) I dont think we need this lemma now.
%\begin{lemma}
%\label{lem:feasible}
%Every feasible matching $M$ in $G$ can be transformed to an $\HH'$-perfect matching in $G'$.
%\end{lemma}
%\begin{proof}
%Map all the non-dummy residents matched to $h$ to its level-$(\ell -1)$ copy if it exists. For every level $0 \le s \le \ell -2$,
%match the hospital $h^s$, to its dummies at level-$s$.
%\end{proof}

%\subsection {Prove Popularity of $map(M')$.}
%Also calls reduction-summary}
%\documentclass[a4paper, UKenglish]{lipics-v2016}
%\begin{document}
\section{Maximum cardinality popular matching}\label{sec:pop}
In this section, we show how to use the reduction in the previous section to compute a maximum cardinality matching that is popular amongst all feasible matchings..
Thus, amongst
all feasible matchings, our algorithm outputs the largest popular matching. We call such a matching 
a {\it maximum cardinality popular matching}.

Our algorithm reduces the HRLQ instance $G$ to an HR instance $G'$ as described in Section~\ref{sec:red}. We then
compute a stable matching $M'$ in $G'$. Finally, to obtain a matching $M$ in $G$ we describe a simple map function.
For every $h \in \HH$, let 
%\begin{eqnarray*}
$M(h) = \RR \cap \left( \bigcup_{s=0}^{\ell-1} M'(h^s) \right)$.
%\end{eqnarray*}
Note that $M(h)$ denotes the set of non-dummy residents matched to any copy $h^s$ of $h$ in $M'$.
Thus, a resident $r$ is matched to a hospital $h$ in $M$ if and only if $r$ is matched to a level-$s$ copy of $h$ in $M'$ for 
some $s \in \{0, \ldots, \ell-1\}$. We say that $M = map(M')$.
We now show some useful invariants about the matching $M = map(M')$. 

\REM{
In the previous section, we saw some properties of a stable matching $M'$ in $G'$. Any matching
$M'$ in $G'$ can be mapped to a matching $M$ in $G$ in a straight forward way: For a hospital $h\in \HH$, let $M'(h)$ be the set of non-dummy residents matched to any copy $h^s$ of $h$ in $M'$. Thus $M'(h)=
\bigcup_{s=0}^{\ell-1} M'(h^s)$. Define $M(h)=M'(h)$. Thus a resident $r$ is matched to a hospital
$h$ in $M$ if and only if $r$ is matched to a level-$s$ copy of $h$ in $M'$ for some $s\in\{0,\ldots,\ell-1\}$. We say that $M=map(M')$.
}

%We show that the residents in $\RR$ and hospitals in $\HH$ can be
%partitioned into levels such that the following invariants hold with respect to $M$.

%\begin{figure}
%\psfrag{Rs}{$\RR_{\ell}$}
%\psfrag{Rt}{$\RR_{\ell-1}$}
%\psfrag{R2}{$\RR_{2}$}
%\psfrag{r1}{$\RR_{1}$}
%\psfrag{R0}{$\RR_{0}$}
%\psfrag{Hs}{$\HH_{\ell}$}
%\psfrag{Ht}{$\HH_{\ell-1}$}
%\psfrag{H2}{$\HH_2$}
%\psfrag{H1}{$\HH_1$}
%\psfrag{H0}{$\HH_0$}
%\includegraphics[scale=0.4]{levels.eps}
%\caption{Partition of the set $\RR$ w.r.t. $M = map(M')$. Bold edges belong to $M$.}
%\end{figure}
%The following invariants for the matching $M = map(M')$ are implied because of invariants of matching $M'$ in $G'$.

\noindent {\bf Division of $\RR$ and $\HH$ into subsets:} 
We divide the residents and hospitals in $G$ into subsets depending upon a matching $M'$ in $G'$.
Let $R_i$ be the set of non-dummy residents matched to a level-$i$ hospital $h^i$ in $M'$. We define
the same set $R_i$ in $G$ as well. Further, define $H_j$ to be the set of hospitals $h\in H$ such
that $ \RR \cap M'(h^j) \neq \emptyset$, that is, level-$j$ copy $h^j$ of $h$ is matched to at least one non-dummy resident in $M'$.  
Define the unmatched residents to be in $R_0$. Also, a non-lower-quota hospital $h$ such that $M(h)=\emptyset$ is defined to be in $H_1$, and a lower-quota hospital $h$ with $M(h)=\emptyset$ is defined to be in $H_{\ell-1}$. 
The following lemma summarizes the properties
of the sets $R_i$ and $H_j$.

\begin{lemma}\label{lem:G-invariants}
Let $M = map(M')$ where $M'$ is a stable matching in $G'$. Then, the following hold:
%Let $M$ be a matching in $G$ such that $M=map(M')$ and $M'$ is a stable matching in $G'$. 
\begin{enumerate}
\item\label{itm:active-hosp} Each hospital is present in at most two sets $H_j,H_{j+1}$ for some $j$. We say that 
$h\in H_j\cap H_{j+1}$.
\item\label{itm:edges} If $h\in H_j\cap H_{j+1}$, then there is no edge from $h$ to any $r\in R_i$
where $i\leq j-1$.
\item\label{itm:undersub} All the non-lower-quota hospitals that are undersubscribed in $M$ are in $H_1$. Moreover, no hospital that is undersubscribed in $M$ is in $H_0$.
\item\label{itm:def} All the deficient
lower-quota hospitals from $M$ are in $H_{\ell-1}$. 
\item\label{itm:undersub-edges}If a non-lower-quota hospital is undersubscribed, it has no edge to any resident in $R_0$.
If a lower-quota hospital is deficient, it does not have an edge to any resident in $R_i$ for $i < \ell-1$.
Similarly an unmatched resident does not have an edge to any hospital in $H_1\cup\ldots\cup H_{\ell-1}$.
\item\label{itm:lq-hosp} Let $h\in\HH$ be such that $|M(h)|>q^-(h)$. Then $h\notin H_2\cup\ldots\cup H_{\ell-1}$. 
\end{enumerate}
\end{lemma}

\begin{proof}
We prove each statement below:
\begin{itemize}
\item{\em Proof of \ref{itm:active-hosp}:} This directly follows from part \ref{inv:three} of Lemma~\ref{lem:SM-invariants}.
\item{\em Proof of \ref{itm:edges}:} Follows from part \ref{inv:steep} of Lemma~\ref{lem:SM-invariants}.
\item{\em Proof of \ref{itm:undersub}:} Let $h$ be undersubscribed in $M$. Thus $|M(h)|<q^+(h)$.
We show that $h\notin H_0$ by showing 
that $h^0$ cannot be active in $M'$. 
For the sake of contradiction, let $(r,h^0)\in M'$. Thus $r\in \langle list_h\rangle$. As $h$ is undersubscribed in $M$, $h^0$ must
be matched in $M'$ to at least one dummy resident in $\D^0_h$, and consequently, $h^1$ must
be matched to at least one dummy resident in $\D^1_h$. But $r$ prefers $h^1$ over $h^0$ in $G'$
and $h^1$ prefers any resident from $\langle list_h\rangle$ over any resident in $\D^1_h$. Thus
$(r,h^1)$ forms a blocking pair w.r.t. $M'$, contradicting the stability of $M'$ in $G'$. 
Hence $h^0$ cannot be active in $M'$ and hence $h\notin H_0$. 

Let $h$ be a non-lower-quota hospital undersubscribed in $M$.
By construction of $G'$, the hospitals $h^2,\ldots,h^{\ell-1}$ have capacity zero and hence cannot be matched to any resident in $M'$. Therefore $h\notin H_2\cup\ldots\cup H_{\ell-1}$. Combining with
the above argument, $h\notin H_0$, and hence $h\in H_1$.
\item{\em Proof of \ref{itm:def}:} 
If a lower-quota hospital $h$ is deficient in $M$, one of its copies is undersubscribed
in $M'$. By part~\ref{inv:two} of Lemma~\ref{lem:SM-invariants}
only $h^{\ell-1}$ can remain undersubscribed in $M'$. Moreover, as $h^{\ell-1}$ is undersubscribed,
no $h^s$, $s<\ell-1$ can be active, otherwise its matched resident creates a blocking pair with $h^{\ell-1}$.

\item{\em Proof of \ref{itm:undersub-edges}:} Let there be an edge $(r,h)$ in $G$ such that $h$ is a 
non-lower-quota hospital undersubscribed in $M$, and let $r\in R_0$. By part \ref{itm:undersub} above,
$h\in H_1$. 
\begin{itemize}
\item If $r$ is unmatched in $M$ and hence in $M'$, then $(r,h^0)$ blocks $M'$. This is
because, since $h^1$ is active in $M'$ implies that $M'(h^0)$ must contain a dummy resident in $\D^0_h$.
But $h^0$ prefers any resident in $\langle list_h \rangle$ over any resident in $\D^0_h$ and $r$ prefers $h^0$ since $r$
in unmatched in $M'$.

\item If $r$ is matched in $M$ and hence in $M'$, it must
be matched to a hospital $\bar{h}\in H_0$ in $M$, and hence to $\bar{h}^0$ in $M'$. But any resident prefers a 
level-$1$ hospital over any level-$0$
hospital in its preference list. Also note that since $h \in H_1$ and is undersubscribed, $M(h^0) = \D^0_h$. Thus,
$M(h^1)$ contains at least one dummy resident from $\D^1_h$. Since $h^1$ prefers $r$ over a dummy resident in $\D^1_h$,
it is clear that $(r,h^1)$ forms a blocking pair with respect to $M'$.
\end{itemize}

A similar argument applies to the case when $h$ is a deficient lower-quota hospital in $M$.

Now let a resident $r$ be unmatched in $M$ and hence in $M'$, and suppose there is an edge $(r,h)$
where $h\in H_i$ for $i>0$. As $h\in H_i$, by part \ref{inv:four-one} of Lemma \ref{lem:SM-invariants},
$M'(h^0)$ contains at least one dummy resident in $\D^0_h$. As $(r,h)$ is an edge in $G$, $r\in \langle list_h \rangle$. But $h^0$
prefers any resident in $\langle list_h \rangle$ over any resident in $\D^0_h$, thus $(r,h^0)$ blocks $M'$ contradicting
its stability.

\item{\em Proof of \ref{itm:lq-hosp}:} For the sake of contradiction, assume that $h\in H_s$ for $s\geq 2$ and still $|M(h)|>q^-(h)$. In this case, $q^-(h)>0$ otherwise in $G'$, $q^+(h^2)=\ldots=q^+(h^{(\ell-1)})=0$ and hence $h\notin H_s$ for any $s\geq 2$. We consider two cases:

{\em Case $1$: $h\in H_s$ for $s\geq 3$:} Consider the matching $M'$ in $G'$. In this case, by part \ref{inv:four-two} of Lemma \ref{lem:SM-invariants}, the level copies $h^0$ and $h^1$ are inactive and $M'(h^i)=\D^i_h$
for $i\in \{0,1\}$. But then, as in the proof of part \ref{inv:one} of Lemma \ref{lem:SM-invariants}, all the
dummy residents must be matched in $M'$. Recall that all the dummy residents in $\D_h$ 
except the first $q^+(h)-q^-(h)$ ones in $\D^1_h$ are top choices of some copy of $h$. 
%Moreover, 
%for each $D^s_h$ except $\D^1_h$, $q^+(h)=|\D^s_h|$. 
Moreover, 
for each $s$ except $s=1$, $q^+(h^s)=|\D^s_h|$. 
Thus the only dummy residents that can
remain unmatched in a stable matching in $G'$ are the first $q^+(h)-q^-(h)$ dummy residents from $\D^1_h$. 
As they are matched in $M'$ by above argument, $|M(h) | \le |\RR \cap M'(h)| \le q_h - |\D_h| \le q^-(h)$. This
contradicts our assumption that $|M(h)| > q^-(h)$.

{\em Case $2$: $h\in H_1\cap H_2$:} Then, in $G'$, $M'(h^0)=\D^0_h$, $M'(h^s)=\D^{(s-1)}_h$
for $3 \le s \le \ell-1 $ by parts \ref{inv:four-two} and \ref{inv:four-three} of Lemma \ref{lem:SM-invariants}
respectively. Thus $M'(h^1)$ has no resident from $\D^0(h)$ and $M'(h^2)$ has no resident from
$\D^2_h$. Let $k$ be the number of non-dummy residents matched to $h^1$ in $M'$. 

If $k\geq q^-(h)$, then at most first $q^+(h)-q^-(h)$ dummy residents from $\D^1_h$ are present in 
$M'(h^1)$. Thus 
all the positions of $h^2$ get matched to residents from $\D^1_h$. Recall that $|\D^1_h|=q^+(h)=q^+(h^1)$
whereas $q^+(h^2)=q^-(h)$ and only the last $q^-(h)$ residents from $\D^1_h$ are present in the
preference list of $h^2$. Thus $h^2$ can not be active in $M'$, contradicting the assumption that
$h\in H_1\cap H_2$. 

Therefore $k< q^-(h)$. Thus exactly $k$ positions of $h^2$ are matched to dummy residents from
$\D^1(h)$. We claim that the remaining $q^+(h^2)-k=q^-(h)-k$ positions of $h^2$ must be matched
to non-dummy residents in $M'$. If not, then $h^2$ must be matched to some dummy resident in $\D^2(h)$,
as $h^2$ is their top choice. This contradicts the above statement that $M'(h^2)$ has no resident
from $\D^2_h$. Thus total number of non-dummy residents matched to $h$ in $M$ is $k+q^-(h)-k=q^-(h)$
contradicting the assumption that $|M(h)|>q^-(h)$.     
\end{itemize}
\end{proof}

\noindent Throughout the following discussion, assume that $M$ is a matching which is a map of a stable
matching $M'$ in $G'$ and $N$ is any feasible matching in $G$. We prove below that $M$ is in fact feasible in $G$.
\begin{theorem}\label{thm:feasible}
If $G$ admits a feasible matching, then $M = map(M')$ is feasible for $G$. 
\end{theorem}
\begin{proof}
%Let $M=map(M')$ where $M'$ is a stable matching in $G'$. 
Suppose $M$ is not feasible. Thus, there
is a deficient lower-quota hospital $h$ in $M$. Let $N$ be a feasible matching in $G$. 
Consider decomposition of $M\oplus N$ into (possibly non-simple) paths and cycles as described in Section \ref{sec:pop-def}.
As $h$ is deficient in $M$ and not deficient in $N$, there must be a path $\rho$ in $M\oplus N$
ending in $h$. Moreover, if the other end of $\rho$ is a hospital $h'$ then $|M(h')|-|N(h')|>0$. Note that in this case, $\rho$ has even-length and hence ends with a $M$-edge. The other case is where $\rho$ ends in a resident $r$ and hence ends with a $N$-edge. We consider the
two cases below:

\begin{itemize}
\item {\bf $\rho$ ends in a hospital $h'$:} As $h$ is deficient in $M$,
$h\in H_{\ell-1}$ by part \ref{itm:def} of Lemma \ref{lem:G-invariants}. 
Also, since $|M(h')|>|N(h')|\geq q^-(h')$, by part \ref{itm:lq-hosp} of Lemma \ref{lem:G-invariants},
$h'\in H_0\cup H_1$. Thus $\rho$ starts at $H_{\ell-1}$ and ends in $H_0$ or $H_1$. Let 
$\rho=\langle h,r_1,h_1,r_2,h_2,\ldots,r_t, h_t, r',h'\rangle$, where $(r_i,h_i)\in M$ and $(r',h')\in M$. 
We show below that such a path $\rho$ can not exist and hence $M$ must be feasible.

By part \ref{itm:undersub-edges} of Lemma \ref{lem:G-invariants}, $h$ has edges only to residents in $R_{\ell-1}$. Hence $r_1\in R_{\ell-1}$ and hence $h_1\in H_{\ell-1}$. By part \ref{itm:edges} of Lemma \ref{lem:G-invariants}, $h_1$ has no edges to residents in $R_0\cup\ldots\cup R_{\ell-3}$. Therefore $r_2\in R_{\ell-1}\cup R_{\ell-2}$ and $h_2\in H_{\ell-1}\cup H_{\ell-2}$. Thus each $h_i\in \rho$ can not be in
$H_j$, for any $j<\ell-i$.  But $h'\in H_0\cup H_1$ and hence $r'\in R_0\cup R_1$. Therefore
$h_t\notin H_3\cup\ldots\cup H_{\ell-1}$ by part \ref{itm:edges} of Lemma \ref{lem:G-invariants}, otherwise
$(h_t,r')$ edge can not exist in $G$. In other words, $\rho$ has to contain at least one hospital from each level $i$, $1\leq i \leq \ell-1$. Thus $t\geq \ell-2$. Moreover, all the hospitals in $\rho$ which 
are in $H_{\ell-1}\cup\ldots \cup H_2$ are lower-quota hospitals. Thus $\rho$ has at least $t+1=\ell-1$ lower-quota hospitals. Note that this count includes repetitions, as a hospital can appear multiple times
in $\rho$. However, any hospital in $H_2\cup\ldots\cup H_{\ell-1}$ can not be matched to more than
$q^-(h)$ residents in $M$ by part \ref{itm:lq-hosp} and hence can appear at most $q^-(h)$ times on $\rho$. But then the sum
of lower quotas of all the hospitals is $\ell-2$, contradicting that $\rho$ has a total of $\ell-1$ occurrences of lower-quota hospitals. Thus such a path $\rho$ can not exist and $M$ must be feasible.

\item {\bf $\rho$ ends in a resident $r$:} Now consider the case where $\rho$ ends at a resident $r$. Then the last edge on $\rho$ must be a $N$-edge and hence $r$ is unmatched in $M$. Therefore $r\in R_0$. Let $\rho=\langle h,r_1,h_1,r_2,h_2,\ldots,r_t, h_t, r\rangle$ where $(r_i,h_i)\in M$ for $1\leq i\leq t$ and the remaining
edges are in $N$. Consider the first hospital, say $h_j$ on $\rho$ such that $h_j\in H_2$ and for each $h_i, i<j$, $h_i\in H_3\cup\ldots\cup H_{\ell-1}$. Such an $h_j$ has to exist by the argument given
for the previous case. Moreover, $j\geq \ell-2$ as $\rho$ has to contain at least one hospital from each
level as described in the previous case. Thus the number of occurrences of lower-quota hospitals on
$\rho$ exceeds the sum of lower quotas and hence such a $\rho$ can not exist.
\end{itemize}
This completes the proof of the lemma.
\end{proof}

In Lemma \ref{lem:label} and Theorem \ref{thm:levels} below, we give crucial properties of the division of $\RR$ and $\HH$ that will be helpful in proving popularity of the matching $M$ which is a map of a stable matching $M'$ in $G'$. 
\begin{lemma}\label{lem:label}
Let $N$ be any feasible matching. Let $(r,h)\in M$ and $(r',h)\in N$ such that $r'=$\cor$_h(r, M, N)$.
Further let $h\in H_j\cap H_{j+1}$ and $r\in R_{j+1}$. Further, let $r'\in R_j$. Then the label on
$(r',h)$ edge is $(-1,-1)$. 
\end{lemma}

\begin{proof}
Clearly $r'$ is not matched to $h$ in $M$, as $\cor(r)$ is picked only from $M(h)\oplus M'(h)$
and $r'\in M'(h)$. Let $r'\in M(h')$. Assume, for the sake of contradiction, that $(r',h)$ does not have
$(-1,-1)$ label. Consider the same edge in $G'$. In $G'$, this is an edge between a $j$-level resident
and a $j+1$-level hospital.
The label on $(r',h)$ in $G$ can not be $(\times,1)$, as this cause the label on the $(r',h^{j+1})$
edge in $G'$ to be $(1,1)$. This is because, if $h$ prefers $r'$ over $r$ in $G$, the preference remains same in $G'$ as well. On the other hand, $r'$ prefers any $j+1$-level hospital over any $j$-level hospital, and hence $h^{j+1}$ over $h'^j$. So it must be the case that the label on $(r',h)$ in $G$ must be 
$(1,-1)$. But in this case, in $G'$, $h^j$ is matched to one of its last dummies since $h^{j+1}$ is active.
Thus $(r',h^j)$ forms a blocking pair with respect to $M'$ in $G'$. This proves that the label on $(r',h)$
in $G$ must be $(-1,-1)$. 
%As there is an edge between $h$ and $r'$ in $G$, $h^j$ must have proposed
%$r'$ and $r'$ must have rejected $h$ during the course of the algorithm. Thus $r'$ prefers
%$h'$ over $h$. Similarly, if $h^{j+1}$ had proposed $r'$, $r'$ would have accepted since $r'$ prefers
%any level-$j+1$ hospital over a level-$j$ hospital. This means $h^{j+1}$ never proposed $r'$. Since
%$h^{j+1}$ proposed $r$, $r$ must be a more preferred resident for $h$ than $r'$. Hence the label
%on $(r',h)$ must be $(-1,-1)$.
\end{proof}

Let $\rho$ be a path in $M\oplus N$ where $M$ is the map of a stable matching $M'$ in $G'$ and $N$ is any feasible matching in $G$. 
Here $\rho$ is constructed according to the decomposition described in Section~\ref{sec:pop-def}. 
Furthermore, the labels on edges of $N \setminus M$ are assigned as described in Section~\ref{sec:pop-def}. 
The following theorem is similar to the one proved in \cite{Kavitha14} for the stable marriage setting.
We adapt the proof here for our setting.
\begin{theorem}\label{thm:levels}
Let $\rho=\langle h_0,r_1,h_1,r_2,h_2,\ldots,h_t, r_{t+1}\rangle$. Moreover, let 
$h_0\in H_p\cap H_{p+1}$ and $r_{t+1}\in R_q$. Then the number of $(1,1)$ edges in $\rho$ is at most
the number of $(-1,-1)$ edges plus $q-p$. Thus $(r_k,h_k)\in M$ for all $k$ and $(h_k,r_{k+1})\in N$ with $r_{k+1}=$\cor$_{h_k}(r_k, M, N)$.
\end{theorem}

\begin{proof}
We prove this by induction on the number of $(-1,-1)$ edges. Note that, except $h_0$, all the $h_i$s
are matched in $M'$, and hence we can consider them at the same level as their matched residents.

Base case: Let $\rho$ have no $(-1,-1)$ edges. As $\rho$ starts at $h\in H_p\cap H_{p+1}$, $r_1$ has to
be in level-$p+1$ or above. This is because there is no edge from $h$ to a resident in $R_0\cup\ldots\cup R_{p-1}$, and if $r_1\in R_p$ then the label on $(h_0,r_1)$ must be $(-1,-1)$ by Lemma \ref{lem:label}.
By assumption, there is no $(-1,-1)$ edge in $\rho$. So $r_1\in R_{j}$ for some $j$, $p+1\leq j\leq \ell$.
Therefore $h_1\in H_{j}$. 

Thus the path can only use edges from a hospital at a lower level to a resident at the same or higher  level. Further, there is no $(1,1)$ edge in $H_k\times R_k$ for any $k$. So $(1,1)$ edges can 
appear in $\rho$ only when it goes from a hospital in a lower level to a resident in a higher level. So
there can be at most $q-p$ many $(1,1)$ edges on $\rho$.

Induction step: Let the theorem hold for at most $i-1$ many $(-1,-1)$ edges. Let $(h_k,r_{k+1})$ be
one such edge. Further, let $h_k\in H_a$ and $r_{k+1}\in R_b$. Consider the two subpaths
$\rho_1=\langle h_0,\ldots,r_k\rangle$ and $\rho_2=\langle h_{k+1}\ldots,r_{t+1}$. As the number of
$(-1,-1)$ edges in each of $\rho_1$
and $\rho_2$ is less than $i$, the induction hypothesis holds. Therefore, the number of $(1,1)$ edges
in $\rho_1$ is at most $a-p$ plus the number of $(-1,-1)$ edges in $\rho_1$. Similarly, the number
of $(1,1)$ edges in $\rho_2$ is $q-b$ plus the number of $(1,1)$ edges in $\rho_2$. The number of $(-1,-1)$ edges in $\rho$ is one more than the total number of $(-1,-1)$ edges in $\rho_1$ and $\rho_2$. Hence the number
of $(1,1)$ edges in $\rho$ is at most the number of $(-1,-1)$ edges in $\rho$ plus $a-p+q-b-1$.
As there is an edge between $h_k$ and $r_{k+1}$, $b\geq a-1$ by Invariant $2$. Thus $a-p+q-b-1\leq
q-p$, which completes the proof.
%Let $\rho$ be as mentioned
%above and let $\rho$ have $i$ edges labelled $(-1,-1)$. Further, let $h
\end{proof}

The following theorem shows that $M$
is a popular matching amongst all the feasible matchings in $G$. 
\begin{theorem}\label{thm:pop}
Let $N$ be
any feasible matching in $G$.
\begin{enumerate}
\item If $\rho$ is an alternating cycle in the decomposition of $M\oplus N$, then $\Delta(M\oplus \rho,M)\leq \Delta(M,M\oplus \rho)$.
\item If $\rho$ is an alternating path in the decomposition of $M\oplus N$ with exactly one end-point matched in $M$,
then $\Delta(M\oplus \rho,M)\leq \Delta(M,M\oplus \rho)$.
\item If $\rho$ is an alternating path in the decomposition of $M\oplus N$ with both the end-points matched in $M$
then $\Delta(M\oplus \rho,M) \leq \Delta(M,M\oplus \rho)$.
\end{enumerate}
\end{theorem}
\begin{proof}
\begin{enumerate}

\item Let $\rho$ be an alternating cycle in $M\oplus N$. Further, let $(r,h)\in M$. Consider 
$\rho'=\rho\setminus \{(r,h)\}$ is an alternating path from $h$ to $r$ which starts and ends at the
same level. Hence the number of $(1,1)$ edges on $\rho'$ is at most the number of $(-1,-1)$ edges
on $\rho'$. The same holds for $\rho$.

\item Let $\rho$ be an alternating path in $M\oplus N$ with exactly one end-point matched in $M$.
Thus $\rho$ has even length, and both its end-points are either hospitals or both are residents.

Consider the first case. So let $\rho=\langle h_0, r_1,h_1,\ldots,r_t,h_t\rangle$ where $(r_i,h_i)\in M$
for all $i$. Thus $|M(h_0)|<|N(h_0)|\leq q^+(h_0)$, and hence $h_0$ is under-subscribed.
Then by part \ref{itm:undersub} of Lemma \ref{lem:G-invariants} and feasibility of $M$, $h_0\notin H_0$. By feasibility of $N$, $h_t\in
H_0\cup H_1$. As $(r_t,h_t)\in M$, $r\in R_0\cup R_1$ by the definition of levels. 
Consider the subpath $\rho'=\rho\setminus \{(r_t,h_t)\}$ i.e. the path obtained by removing the
edge $(r_t,h_t)$ from $\rho$. Applying Theorem \ref{thm:levels} to $\rho'$ with $p\geq 1$ and $q=0$ or $q=1$, we get that the number of $(1,1)$ edges on $\rho'$ is at most the number of $(-1,-1)$ edges
on $\rho'$.

Consider the case when both the end-points of $\rho$ are residents. Thus $\rho=\langle r_0,h_1,r_1,
\ldots, h_t,r_t\rangle$ where $(h_i,r_i)\in M$ for all $i$. Again consider $\rho'=\rho\setminus \{(h_t,r_t)\}$. As $r_0$ is unmatched in $M$, $r_0\in R_0$ by the definition of levels. Applying Theorem \ref{thm:levels} to $\rho'$ with $q=0$, we get that the number of $(1,1)$ edges on $\rho'$ is at most the number of $(-1,-1)$ edges
on $\rho'$.

\item Consider the case when both the end-points of the alternating path $\rho$ are matched in $M$.
Thus one end-point of $\rho$ is a hospital whereas the other end-point is a resident. Let $\rho=\langle r_0,h_0, \ldots,r_t,h_t\rangle$ where $(r_i,h_i)\in M$ for all $i$. Hence $|M(h_t)|>|N(h_t)|\geq q^-(h_t)$ by feasibility of $N$. Therefore $h_t\in H_0\cup H_1$ by part \ref{itm:lq-hosp} of Lemma \ref{lem:G-invariants} which implies that $r_t\in R_0\cup R_1$. Consider the subpath
$\rho'=\rho\setminus \{(r_0,h_0),(r_t,h_t)\}$. Thus $\rho'$ begins at $h_0$ and ends at $r_t$. Applying
Theorem \ref{thm:levels} with $q=1$ and $0\leq p\leq \ell$ gives that the number of $(1,1)$ edges
on $\rho'$, and hence on $\rho$, is at most one more than the number of $(-1,-1)$ edges on $\rho$.
These votes in favor of $N$ are compensated by the end-points $r_0$ and $h_t$ as $r_0$ is
unmatched in $N$ and $|M(h_t)|>|N(h_t)|$. 
\end{enumerate}
This completes the proof of the theorem.
\end{proof}

The following lemma proves that $M$ is a maximum cardinality popular matching in $G$. 
\begin{lemma}\label{lem:maxpop}
For any feasible matching $N$ in $G$ such that $|N|>|M|$, $\Delta(N,M)<\Delta(M,N)$.
\end{lemma}
\begin{proof}
Consider $M\oplus N$. There is an alternating path $\rho$ in $M\oplus N$ such that $\rho$
has both the end-points unmatched/undersubscribed in $M$ and the path begins and ends with edges in $N$. Let $\rho=\langle h_0,r_1,h_1,\ldots,r_t,h_t,r_{t+1}\rangle$
where $(r_i,h_i)\in M$ for each $i$. As $h_0$ is under-subscribed in $M$, and $r_{t+1}$ is unmatched in $M$, by the definition of levels and part \ref{itm:undersub} of Lemma \ref{lem:G-invariants}, 
$h_0\in \bigcup_{j=1}^{(\ell-1)}H_j$ and $r_{t+1}\in R_0$. Further, by
part \ref{itm:undersub-edges} of Lemma \ref{lem:G-invariants}, there is no edge from $h_0$ to any $r\in R_0$ and no edge from $r_{t+1}$ to any $h\in 
H_1\cup\ldots\cup H_{\ell-1}$.
%\textcolor{magenta}{TODO check Add this to the invariants, and of course, prove the invariants!!} 
The path $\rho$ begins in $\bigcup_{j=1}^{(\ell-1)} H_j$ and has to end in $R_0$, and the only edges to $R_0$ are from vertices in $H_0\cup H_1$. Further, each $r_i$, $i\leq t$ is matched in $M$ and hence the corresponding hospital $h_i$ is at the same level as $r_i$. By part \ref{itm:edges} of Lemma \ref{lem:G-invariants}, if $h_i\in H_2\cup\ldots\cup H_{(\ell-1)}$ then $r_{i+1}\notin R_0$ for any $i$. Therefore there
must be an edge $(h_i,r_{i+1})$ on $\rho$ such that $h_i\in H_1$ and $r_i\in R_0$ and $1\leq i\leq t-1$.
By Lemma \ref{lem:label}, this edge must be labelled $(-1,-1)$. Now consider the two subpaths
$\rho_1=\langle h_0,r_1,h_1,\ldots,r_i\rangle$ and $\rho_2=\langle h_{i+1},r_{i+2},\ldots,h_t,r_{t+1}\rangle$.
These are the subpaths obtained by removing the subpath $r_i, h_i,r_{i+1},h_{i+1}$ from $\rho$. By assumption,
$h_i\in H_1$ and $r_{i+1}\in R_0$, hence $r_i\in R_1$ and $h_{i+1}\in H_0$. Therefore, applying
Theorem \ref{thm:levels} to $\rho_1$ and $\rho_2$ gives that the number of $(1,1)$ edges on $\rho_1$
and $\rho_2$ are at most the number of $(-1,-1)$ edges on them.  Thus $N$ does not get more votes than $M$ on $\rho_1$ or $\rho_2$. Further, $M$ gets two more votes on the $(h_i,r_{i+1})$ edge.
Hence the lemma follows.
\end{proof}

\section{Popular matching amongst maximum cardinality feasible matchings}\label{sec:popmax}
In this section, we modify the reduction in Section~\ref{sec:red} to obtain a matching
that is popular amongst all the maximum cardinality feasible matchings of the HRLQ instance.
%It is easy to see that a maximum cardinality feasible matching is also a maximum cardinality matching,
%thus the matching constructed in this section is also a maximum cardinality matching.
The reduction is very similar to the one described in Section~\ref{sec:red} except for the number
of copies of each hospital. The HR instance $G'$ described in Section~\ref{sec:red},
has $\ell=2+\sum_{h\in \HH}q^-(h)$ copies corresponding to each hospital in $G$. 

\subsection{Reduction to HR instance}\label{sec:max-red}
%Here the HR instance $G''$ has $\ell=|\RR|+\sum_{h\in \HH}q^-(h)$ copies corresponding to each hospital in $G$. We describe the reduction below:

Given HRLQ instance $G=(\RR\cup\HH,E)$, the corresponding
HR instance $G''=(\RR''\cup\HH'',E'')$ is as follows. 
We set $\ell=|\RR|+\sum_{h\in \HH} q^-(h)$.
We start with the vertices in $G''$.

\noindent {\bf \underline {The set $\HH''$:}} For every hospital $h \in \HH$, we have $\ell$ copies of 
$h$ in $\HH''$. The set $\HH''$ and the capacities are as given below.
\begin{eqnarray*}
\HH'' & = & \{h^0,\ldots, h^{\ell-1}\mid h\in \HH\}\\
\textrm{Capacities of $h \in \HH''$: }q^+(h^s) & = & q^+(h), \quad s\in \{0,\ldots, |\RR|-1\}\\
q^+(h^s) & = & q^-(h), \quad s\in \{|\RR|,\ldots,\ell-1\}\\
\end{eqnarray*}
The following observation is immediate.
\begin{invariant}
\label{inv:sum-cap-H2}
For a hospital $h \in \HH$, the sum of capacities of all level copies of $h$ in $G''$ is $q_h = |\RR| \cdot q^+(h) + (\ell - 2) \cdot q^-(h)$.
\end{invariant}

\noindent {\bf \underline{The set $\RR''$:}} The set of residents $\RR''$ consists of the set $\RR$ along
with a set of dummy residents $\mathcal{D}_h$ corresponding to every hospital $h \in \HH$.
The set $\RR''$ and $\mathcal{D}_h$ are as defined below:
\[
\begin{array}{rcl}
\RR'' & = & \RR\cup \left(\bigcup_{h\in \HH} \D_h \right)\textrm{ where } \ \ \ \ \ \ \ \ 
\quad \D_h  = \bigcup_{s\in\{0,\ldots,\ell-1\}}\D^s_h \ \ \ \ \  \  \forall h \in \HH \\
\textrm{where }\D^s_h& = & \{d^s_{h,1}\ldots,d^s_{h,q^+(h)}\}, \quad  \quad s\in\{0,|\RR|-1\}\\
\textrm{and }\D^s_h&=&\{d^s_{h,1}\ldots,d^s_{h,q^-(h)}\}, \quad  \quad {s\in \{|\RR|,\ldots,\ell-2\}}
%\bigcup_{s\in \{0,1\}} \{d^s_{h,1}\ldots,d^s_{h,q^+(h)}\} \cup 
%\bigcup_{s\in \{2,\ldots,\ell\}} \{d^s_{h,1}\ldots,d^s_{h,q^-(h)}\} \\
%\textrm{Also define }D_h & = & \bigcup_{s=0}^\ell D^s_h\\
\end{array}\]
Here $\D_h$ is the set of {\em dummy residents corresponding to $h$} and $\D^s_h$ is the set of {\em level-$s$
dummy residents corresponding to $h$}.
Note that $\D^s_h=\emptyset$ for each $s\in\{|\RR|,\ldots,\ell-1\}$ if $q^-(h)=0$.
\noindent The following observation captures the number of dummy residents corresponding to every hospital $h \in \HH$.
\begin{invariant}
\label{inv:sum-dummy-H2}
For a hospital $h \in \HH$, the total number of dummy residents corresponding $h$ in $\RR''$ is $|\D_h| = |\RR| \cdot q^+(h) + (\ell - 3) \cdot q^-(h)$.
\end{invariant}
\vspace{0.1in}
\noindent {\bf \underline{Preference lists:}}
Recall that the preference list of a resident $r$ in $G$ is denoted by $\langle list_r\rangle$ and that
of a hospital $h$ is $\langle list_h\rangle$.
%Preference lists in $G''$ are given by

\vspace{0.1in}
\noindent {\em Hospitals' preference lists:} For $h^s \in \HH''$:
%{\em Hospitals' preference lists:}
\[
\begin{array}{lcl}
s=0& :& \langle list_h\rangle  \circ \langle  \D^0_h \rangle\\
%\textrm{For }h^s\in \HH' , & & \\
s\in \{1,\ldots, |\RR|-2\} &:& \langle D^{s-1}_h\rangle \circ \langle list_h\rangle \circ \langle D^s_h\rangle\\
s=|\RR|-1 &:&  \langle d^{|\RR|-2}_{h,k},\ldots, d^{|\RR| -2}_{h,q^+(h)}\rangle \circ \langle list_h\rangle \circ \langle \D^{|\RR| -1}_h\rangle, \qquad k=q^+(h)-q^-(h)+1\\
s \in \{ |\RR|, \ldots, \ell -2\}&:& \langle D^{s-1}_h\rangle \circ \langle list_h\rangle \circ \langle D^s_h\rangle\\
s=\ell-1 &:& \langle \D^{(\ell-1)}_h\rangle\circ\langle list_h\rangle
\end{array}
\]
\vspace{0.1in}
\noindent {\em Residents' preference lists:}
\[
\begin{array}{lcl}
\textrm{For }r\in \RR &:& \langle list_r \rangle^{(\ell-1)} \circ\ldots\circ\langle list_r\rangle^0\\
\textrm{For }h\in \HH,  \ \ \ d^s_{h,i}\in \D_h &:& \\
 s\in \{0,\ldots,\ell-2\}\setminus \{|\RR|-1\} &:& \langle h^s, h^{s+1} \rangle \\ %& & \textrm{if }s\in \{0,\ldots,\ell-2\}\setminus \{|\RR|-1\}\\
s=|\RR|-1, \ \  i\in \{1,\ldots,q^+(h)-q^-(h)\} &:& \langle h^{(|\RR|-1)}\rangle \\ % & & \textrm{if }s=|\RR|-1, i\in \{1,\ldots,q^+(h)-q^-(h)\}\\
s=|\RR|-1, \ \ i\in \{q^+(h)-q^-(h)+1,\ldots,q^+(h)\} &:& 
\langle h^{(|\RR|-1)}, h^{|\RR|}\rangle% & & \textrm{if } s=|\RR|-1, i\in \{q^+(h)-q^-(h)+1,\ldots,q^+(h)\}
\end{array}\]
%Here $\circ$ denotes concatenation and $\langle D^s_h\rangle$ denotes the list consisting of
%elements of $D^s_h$ in increasing order of indices. Also, $\langle list_r\rangle^s$ is same as
%$\langle list_r\rangle$ with each hospital $h$ replaced by $h^s$.  
%\begin{eqnarray*}

\noindent The following lemma summarizes properties of a stable matching $M''$ in $G''$. It is an analogue of
Lemma~\ref{lem:SM-invariants} from Section~\ref{sec:red}. The proof is 
similar to the proof of Lemma~\ref{lem:SM-invariants}; we omit it here.
\begin{lemma}\label{lem:SM-max-invariants}
\begin{enumerate}
The following hold for any stable matching $M''$ in $G''$:
\item\label{inv:max-one} For any $h \in \HH$, $M''$ matches at most $q^+(h)$ non-dummy residents across all its 
level copies in $G''$.
\item \label{inv:max-two}The matching $M''$ in $G''$ leaves only the level-$(\ell -1)$ copy of any hospital (if it exists) under-subscribed. 
\item Let $h^s \in \HH''$ be active in $M''$.  Then,
\begin{enumerate}
\item\label{inv:max-four-one} At least one position of $h^{s-1}$ is matched in $M''$ to a dummy resident at level-($s-1$).
\item\label{inv:max-four-two} For $0 \le j \le s-2$, $h^j$ is inactive in $M''$ and all positions of $h^{j}$ are matched to dummy residents of level-($j$).
\item\label{inv:max-four-three} For $s+2 \le j \le \ell-1$, $h^j$ is inactive in $M''$ and all positions of $h^{j}$ are matched to  dummy residents of level-($j-1$).
\end{enumerate} 
\item\label{inv:max-three} For any $h \in \HH$, at most two consecutive level copies $h^s$ and $h^{s+1}$ are active in $M''$.
\item\label{inv:max-steep} A level-$s$ resident $r$ in $M''$ does not have any hospital $h$ in its preference list which is active at level-$(s+2)$ in $M''$.
%\textcolor{blue}{Proof: Otherwise the edge $(r, h^{s+1})$ forms a blocking pair with respect to $M'$.}
\end{enumerate}
\end{lemma}

\subsection{Popularity of the matching}
In order to compute a matching that popular amongst all the maximum cardinality feasible matchings,
we execute the following algorithm. Construct the graph $G''$ and compute a stable matching $M''$ 
in $G''$. Because of the invariants on $M''$, there exists a natural map from $M''$ to a matching $M$ in $G$.
Let $M(h) = \RR \cap \left(\bigcup_{s=0}^{\ell-1} M''(h^s)\right)$.
We prove that $M=map(M'')$  is popular amongst all maximum cardinality feasible matchings
in $G$.
%if $M''$ is a stable matching in $G''$. 

\noindent {\bf Division of residents and hospitals into subsets:} As in Section~\ref{sec:pop}, we divide residents and hospitals into subsets based on a stable matching $M''$ in $G''$. Thus $R_i$ is
the set of non-dummy residents matched to a level-$i$ hospital $h^i$ in $M''$, $H_i$ is the set of
hospitals that are active at level-$i$ in $M''$, unmatched residents are in $R_0$. Also, if 
$M(h)=\emptyset$ then $h\in H_{(|\RR|-1)}$ if $q^-(h)=0$ and $h\in H_{\ell-1}$ if $q^-(h)>0$.
The following lemma summarizes the properties of a matching $M$ in $G$ where $M=map(M'')$
and $M''$ is a stable matching in $G''$.
\begin{lemma}\label{lem:max-G-invariants}
Let $M$ be a matching in $G$ such that $M=map(M'')$ and $M''$ is a stable matching in $G''$. 
\begin{enumerate}
\item\label{itm:max-active-hosp} Each hospital is present in at most two sets $H_j,H_{j+1}$ for some $j$. We say that 
$h\in H_j\cap H_{j+1}$.
\item\label{itm:max-edges} If $h\in H_j\cap H_{j+1}$, then there is no edge from $h$ to any $r\in R_i$
where $i\leq j-1$.
\item\label{itm:max-undersub} All the non-lower-quota hospitals that are undersubscribed in $M$ are in $H_{(|\RR|-1)}$. 
\item\label{itm:max-def} All the deficient
lower-quota hospitals from $M$ are in $H_{(\ell-1)}$. 
\item\label{itm:max-undersub-edges}If a non-lower-quota hospital is undersubscribed, it has no edge to any resident in $R_0$.
If a lower-quota hospital is deficient it does not have an edge to any resident in $R_i$ for $i\leq \ell-1$.
Similarly an unmatched resident does not have an edge to any hospital in $H_1\cup\ldots\cup H_{\ell-1}$.
\item\label{itm:max-lq-hosp} Let $h\in\HH$ be such that $|M(h)|>q^-(h)$. Then $h\notin H_{|\RR|}\cup\ldots\cup H_{(\ell-1)}$. 
\end{enumerate}
\end{lemma}

\noindent In Theorem \ref{thm:max-feasible} below, we prove that $M=map(M'')$ is feasible in $G$ if $M''$ is
stable in $G''$. 
%Moreover, in Theorem \ref{thm:max-pop}, we prove that $M$ has maximum cardinality
%and it is popular amongst all the maximum cardinality feasible matchings in $G$.
\begin{theorem}\label{thm:max-feasible}
If $G$ admits a feasible matching then the map $M$ of any stable matching $M''$ in $G''$ is feasible.
\end{theorem}
\begin{proof}
The proof is analogous to that of Theorem \ref{thm:feasible}. Assume, for the sake of contradiction, that $M$ is not feasible and hence there is a hospital $h$ such that $|M(h)|<q^-(h)$. We need to consider $M\oplus N$ where $N$ is a feasible matching in $G$. As $N$ is feasible, $|N(h)|\geq q^-(h)$ and hence there must be a path $\rho$ in the decomposition of $M\oplus N$ with $h$ as one of its end-points. 
The first case is that the other end-point of $\rho$ is  a hospital $h'$ and hence the last edge of $\rho$ must be a $M$-edge (call this Case~1). In the second case, the other end-point of $\rho$ is a 
a resident $r$, implying last edge of $\rho$ is a $N$-edge (call this Case~2). 

Consider Case~1. As the path ends at $h'$ with an $M$-edge, $|M(h')|>|N(h')|\geq q^-(h')$.
Then, by part \ref{itm:max-lq-hosp} of Lemma~\ref{lem:max-G-invariants}, $h'\in H_0\cup\ldots\cup H_{(|\RR|-1)}$. Also, by part \ref{itm:def} of Lemma~\ref{lem:max-G-invariants}, $h\in H_{(\ell-1)}$. Then,
by a similar argument as in the proof of Theorem~\ref{thm:feasible}, the length of $\rho$ exceeds the sum of lower-quotas of all the hospitals and hence such a path can not exist.

When $\rho$ ends in a resident $r$, and hence the last edge is a $N$-edge, $r$ is unmatched in $M$
and hence must be in $R_0$. A similar argument applies in this case as well, for the sub-path of $\rho$ where we consider the first hospital $h_j$ such that $h_j\in H_{|\RR|}$ and for each $h_i, i<j$, $h_i\in H_{(|\RR|+1)}\cup\ldots\cup H_{(\ell-1)}$. We omit the details.
\end{proof}

We need the following lemma and theorem for proving popularity of $M$. They are analogues of
Lemma~\ref{lem:label} and Theorem~\ref{thm:levels} respectively from Section~\ref{sec:pop}.
\begin{lemma}\label{lem:max-label}
Let $M_1$ be any feasible matching. Let $(r,h)\in M$ and $(r',h)\in M_1$ such that $r'=$\cor$_h(r, M, M_1)$.
Further let $h\in H_j\cap H_{j+1}$ and $r\in R_{j+1}$. Further, let $r'\in R_j$. Then the label on
$(r',h)$ edge is $(-1,-1)$. 
\end{lemma}

Let $\rho$ be a path in $M\oplus M_1$ where $M$ is the map of a stable matching $M''$ in $G''$ and $M_1$ is any feasible matching in $G$. 
Here $\rho$ is constructed according to the decomposition described in Section~\ref{sec:pop-def}. 
The labels are assigned according to the voting described in Section~\ref{sec:pop-def}. 

\begin{theorem}\label{thm:max-levels}
Let $\rho=\langle h_0,r_1,h_1,r_2,h_2,\ldots,h_t, r_{t+1}\rangle$. Moreover, let 
$h_0\in H_p\cap H_{p+1}$ and $r_{t+1}\in R_q$. Then the number of $(1,1)$ edges in $\rho$ is at most
the number of $(-1,-1)$ edges plus $q-p$. Thus $(r_k,h_k)\in M$ for all $k$ and $(h_k,r_{k+1})\in M_1$ with $r_{k+1}=$\cor$_{h_k}(r_k, M, M_1)$.
\end{theorem}

Now we prove that $M$ is a maximum cardinality feasible matching in $G$.
\begin{theorem}\label{thm:max-card}
There is no feasible matching $N$ such that $|N|>|M|$.
\end{theorem}
\begin{proof}
For contradiction, assume that there exists a feasible matching $N$ in $G$ such that $|N|>|M|$. Then, in the decomposition of $M\oplus N$, there must be a path with both end-points unmatched / under-subscribed in $M$.
Let $\rho$ be such a path. One end-point of $\rho$ must be a hospital $h$ and other end-point must
be a resident $r$. Then $|M(h)|<|N(h)|$ and hence $h\in \bigcup_{i=|\RR|-1}^{\ell-1} H_{i}$ by part \ref{itm:max-undersub}
of Lemma \ref{lem:max-G-invariants}. Also, $r\in R_0$ by the way division of $\RR$ into subsets is
defined. Let $\rho=\langle r,h_0,r_0,h_1,r_1,\ldots,h_t,r_t,h\rangle$ where $(h_i,r_i)\in M$ for $0\leq i\leq t$. As an unmatched resident has edges only to hospitals in $H_0$ (cf. part \ref{itm:max-undersub-edges} of Lemma \ref{lem:max-G-invariants}), $h_0\in H_0$ and hence $r_0\in R_0$ since $r_0=M(h_0)$. Now again, by part \ref{itm:max-edges} of Lemma \ref{lem:max-G-invariants}, $r_0$
has no edge to any hospital in $H_2\cup\ldots\cup H_{(\ell-1)}$. Hence $h_1\in H_0\cup H_1$. 
Continuing this argument, $h_i\in \bigcup_{j=0}^i H_j$. But then $h$ has no edge to any resident in $\bigcup_{j=0}^{|\RR|-2}$ and so $r_t\in R_{|\RR|-1}$. Therefore $t\geq |\RR|-1$. But then the number 
of residents on $\rho$ is $|\RR|+1$, which is not possible. Hence there is no augmenting path with respect to $M$ in $G$. 
\end{proof}
In the theorem below, we prove that $M$ is popular amongst all the feasible matchings of maximum cardinality.
\begin{theorem}\label{thm:max-pop}
Let $N$ be
any feasible matching in $G$.
\begin{enumerate}
\item If $\rho$ is an alternating cycle in $M\oplus N$, then $\Delta(M\oplus \rho,M)\leq \Delta(M,M\oplus \rho)$.
\item If $\rho$ is an alternating path in $M\oplus N$ with one end-point matched in $M$,
then $\Delta(M\oplus \rho,M)\leq \Delta(M,M\oplus \rho)$.
\end{enumerate}
\end{theorem}
\begin{proof}
Let $\rho$ be an alternating cycle in $M\oplus N$. Further, let $(r,h)\in \rho\cap M$. Then 
$\rho'=\rho\setminus \{(r,h)\}$ is an alternating path from $h$ to $r$ which starts and ends at the
same level. Hence the number of $(1,1)$ edges on $\rho'$ is at most the number of $(-1,-1)$ edges
on $\rho'$. The same holds for $\rho$.

Let $\rho$ be an alternating path in $M\oplus N$ with exactly one end-point matched in $M$.
Thus $\rho$ has even length, and both its end-points are either hospitals or both are residents.

Consider the first case. So let $\rho=\langle h_0, r_1,h_1,\ldots,r_t,h_t\rangle$ where $(r_i,h_i)\in M$
for all $i$. Also, as $N$ is feasible, $q^+(h)\geq |N(h)|>|M(h)|$. By part \ref{itm:max-undersub} of Lemma \ref{lem:max-G-invariants} and feasibility of $M$, $h_0\in \bigcup_{i=|\RR|-1}^{\ell-1} H_i$. By feasibility of $N$, $h_t\in
\bigcup_{i=0}^{|\RR|-1}H_i$. As $(r_t,h_t)\in M$, $r\in \bigcup_{i=0}^{|\RR|-1}$ by the definition of levels. 
Consider the subpath $\rho'=\rho\setminus \{(r_t,h_t)\}$ i.e. the path obtained by removing the
edge $(r_t,h_t)$ from $\rho$. Applying Theorem \ref{thm:max-levels} to $\rho'$ with $p\geq |\RR|-1$ and $q\leq |\RR|-1$, we get that the number of $(1,1)$ edges on $\rho'$ is at most the number of $(-1,-1)$ edges on $\rho'$.

Consider the case when both the end-points of $\rho$ are residents. Thus $\rho=\langle r_0,h_1,r_1,
\ldots, h_t,r_t\rangle$ where $(h_i,r_i)\in M$ for all $i$. Again consider $\rho'=\rho\setminus \{(h_t,r_t)\}$. As $r_0$ is unmatched in $M$, $r_0\in R_0$ by the definition of levels. Applying Theorem \ref{thm:max-levels} to $\rho'$ with $q=0$, we get that the number of $(1,1)$ edges on $\rho'$ is at most the number of $(-1,-1)$ edges
on $\rho'$.
\end{proof}

\bibliographystyle{abbrv}
\newpage
\bibliography{references}
\end{document}